\documentclass[letterpaper, 10 pt, conference]{ieeeconf}
\IEEEoverridecommandlockouts
\usepackage{amsfonts}
\usepackage{graphicx}
\usepackage{float}
\usepackage{comment}
\usepackage{mathrsfs}
\usepackage{xcolor}

\usepackage{amsthm,amssymb} %para teoremas y demostraciones
\usepackage{amsmath}
\newtheorem{theorem}{Theorem}
\newtheorem{assumption}{Assumption}
\newtheorem{corollary}{Corollary}
\newtheorem{lemma}{Lemma}

\newtheorem*{remark}{Remark}

\newcommand{\x}{\textup{x}}
\newcommand{\y}{\textup{y}}
\newcommand{\w}{\textup{w}}

\def\BibTeX{{\rm B\kern-.05em{\sc i\kern-.025em b}\kern-.08em
    T\kern-.1667em\lower.7ex\hbox{E}\kern-.125emX}}

\begin{document}

\title{\LARGE \bf Estimates for weighted homogeneous delay systems:\\ A Lyapunov-Krasovskii-Razumikhin approach*}

\author{Gerson~Portilla$^{1},$ Irina V. Alexandrova$^{2}$ and Sabine~Mondié$^{1}$% <-this % stops a space
\thanks{$^{1}$Gerson~Portilla and Sabine~Mondié are with the Department of Automatic Control, CINVESTAV-IPN, 07360 Mexico D.F., Mexico {\tt\small \{gportilla,smondie\}@ctrl.cinvestav.mx}}%
\thanks{$^{2}$Irina V. Alexandrova is with the Department of Control Theory, St. Petersburg State University, 7/9 Universitetskaya nab., St. Petersburg, 199034, Russia {\tt\small i.v.aleksandrova@spbu.ru}}%
\thanks{*The work of Gerson Portilla and Sabine Mondié was supported by Projects {\tt\small CONACYT A1-S-24796} and {\tt\small SEP-CINVESTAV 155}, Mexico. The work of Irina Alexandrova was supported by the Russian Science Foundation, Project {\tt\small19-71-00061}.}
}

\maketitle

\begin{abstract}
In this paper, we present estimates for solutions and for the attraction domain of the trivial solution for systems with delayed and nonlinear weighted homogeneous right-hand side of
positive degree. The results are achieved via a generalization of the Lyapunov-Krasovskii functional construction presented recently for homogeneous systems with standard dilation. Along with the classical approach for the calculation of the estimates within the Lyapunov-Krasovskii framework, we develop a novel approach which combines the use of Lyapunov-Krasovskii functionals with ideas of the Razumikhin framework. More precisely, a lower bound for the functional on a special set of functions inspired by the Razumikhin condition is constructed, and an additional condition imposed on the solution of the comparison equation ensures that this bound can be used to estimate all solutions in a certain neighbourhood of the trivial one. An example shows that this approach yields less conservative estimates in comparison with the classical one.
\end{abstract}
%\begin{IEEEkeywords}
%Homogeneous time-delay systems, Lyapunov-Krasovskii functionals, estimates for solutions, attraction region. 
%\end{IEEEkeywords}

% \fi

\section{Introduction}
%%%%%%%%%%%%
When the linear approximation is zero, the homogeneous one can be used
for nonlinear systems analysis and control. Generalised definitions such as weighted homogeneity \cite{rosier1992homogeneous}, \cite{bacciotti2006liapunov} or in the bi-limit \cite{andrieu2008homogeneous} allow covering wider classes of nonlinear systems. The Lyapunov framework has produced significant results on stability \cite{zubov1964methods}, \cite{hermes1991homogeneous}, robustness \cite{rosier1992homogeneous} as well as observer and controller design \cite{hermes1991homogeneous}, to name a few. \\
To study homogeneous systems with delays, researchers have naturally resorted to the Lyapunov-Razumikhin framework \cite{efimov2011homogeneity}. Some general results on delay-independent and finite-time stability for the cases of positive and negative degree, respectively, as well as for stability of locally homogeneous systems are obtained in \cite{efimov2016weighted, efimov2016global}. For weighted homogeneous systems of positive degree, the delay-independent stability was established with the help of the Lyapunov function of the corresponding delay-free system \cite{aleksandrov2014weighted}. The approach has allowed to present the estimates for solutions \cite{aleksandrov2012asymptotic} and for the attraction region, as well as to analyze perturbed systems and complex systems describing the interaction between several subsystems \cite{aleksandrov2014delay}. Moreover, contrary to the often expressed view that the Razumikhin approach allows obtaining qualitative estimates only, it was shown recently \cite{portilla2020comparison} that the estimates of \cite{aleksandrov2012asymptotic} are close enough to the system response. A similar conclusion was made in \cite{efimov2020onestimation} for different kind of systems.\\
Recently, for the case of standard dilation and homogeneity degree strictly greater than one, a Lyapunov-Krasovskii functional was introduced in \cite{Voronezh}, \cite{alexandrova2019lyapunov}. It was inspired by the Lyapunov functional of complete type for delay linear systems \cite{kharitonov2003lyapunov}, \cite{kharitonov2013time}, and lead to stability and instability results \cite{zhabko2020automatica}, estimates of the region of attraction \cite{alexandrova2019lyapunov} and of the system response \cite{portilla2020comparison}, see also \cite{portilla2020thesis}.\\
In this contribution, we extend this functional to the case of weighted homogeneous time-delay systems of positive degree and use it to construct quantitative estimates of the region of attraction and of the system response.
%These estimates are a  valuable measure of the homogeneous system performance,  and can also be used in the design of controller achieving given objectives.
Two approaches are developed. The first one is based on the classical ideas of the Lyapunov-Krasovskii method, whereas a combination of Lyapunov-Krasovskii and Razumikhin techniques is used in the second one. The idea is to construct a lower bound for the Lyapunov-Krasovskii functional which is only valid on a special set of functions inspired by the Razumikhin condition, see \cite{medvedeva2015synthesis}  and \cite{alexandrova2018junction} for the linear and nonlinear cases, respectively.
Exploiting the ideas in \cite{aleksandrov2012asymptotic}, we require the solutions of the comparison equation for the functional to satisfy the same condition, thus ensuring  that the final estimates hold for all solutions from a certain neighbourhood. This approach yields better estimates than the classical one.\\
Note that there exists a parallel work on the generalization of the functional of \cite{Voronezh} to the case of weighted dilation covering both asymptotic stability for the case of positive
degree and boundedness of solutions for that of negative degree \cite{efimov2020analysis}. The main difference between the functional we use and those of \cite{efimov2020analysis} is that to cover both cases the authors of \cite{efimov2020analysis} use a more complex construction with additional parameters, a more complicated bounding and non-standard norms, thus achieving a moderate computational performance. In contrast, we bound the functional and its derivative componentwise following naturally the componentwise definition of homogeneous functions, and use natural norm based on the homogeneous vector norm. Additionally, we present fully computed quantitative estimates of the response and of the attraction region, via this combined Lyapunov-Krasovskii-Razumikhin approach.
%Our computations show that we are able to obtain significantly less conservative bounds than those in \cite{efimov2020analysis} for the positive degree case.
\\
The contribution is organised as follows. Previous results on homogeneous systems are reminded in section II. The Lyapunov-Krasovskii functional construction is presented in Section III.  The functional is applied to the estimation of the attraction region in Section~IV and of the homogeneous system solutions in Section~V. An illustrative example is given in Section~VI. \\
%%%%%%%%%%%%%%%%%%%
\textbf{Notation: } The space of $\mathbb{R}^n$ valued continuous functions on $[-h,0]$ endowed with the norm $ \|\varphi\|_h=\max_{\theta\in[-h,0]}\|\varphi(\theta)\|$ is denoted by $C([-h,0],\mathbb{R}^n)$. Here, $\|\cdot\|$ stands for the Euclidean norm. In computations it turns out to be more convenient to use the following homogeneous norm
$$ \|\varphi\|_\mathscr{H}=\max_{\theta\in[-h,0]}\|\varphi(\theta)\|_{r,p},$$
%%%Notice that the norm is different from the one which is usually used for homogeneous systems!! (supremum in different place)
where $\|\cdot\|_{r,p}$ stands for the typical vector homogeneous norm defined below. The solution of a time delay system and the restriction of the solution to the segment $[t-h,t]$, corresponding to the initial function $\varphi\in C([-h,0],\mathbb{R}^n),$ are respectively denoted by $x(t)$ and $x_t$. If the initial condition is important, we write $x(t,\varphi)$ and $x_t(\varphi)$, respectively.

\section{Preliminaries}
We start with a brief reminder of the definitions related to the homogeneity concept \cite{bacciotti2006liapunov,zubov1962oxford}. Define the vector of weights \mbox{$r=(r_1,\ldots,r_n)^T,$} where $r_i>0,$ $i=\overline{1,n},$ and the dilation operator
\begin{equation*}
    \delta_\varepsilon^r(\x)=(\varepsilon^{r_1}\x_1,\ldots,\varepsilon^{r_n}\x_n),\quad \varepsilon>0.
\end{equation*}
Here, $\x=(\x_1,\ldots,\x_n)^T.$
Then, function
\begin{equation*}
    \|\x\|_{r,p}=\left(\sum_{i=1}^n |\x_i|^{p/r_i}\right)^{1/p},
\end{equation*}
where $p\geq 1,$ is called \textit{the $\delta^r$-homogeneous norm}. Although it is not a norm in the usual sense, it has been shown to be equivalent to the Euclidean norm. A scalar function $V:\mathbb{R}^n\rightarrow \mathbb{R}$ is called \textit{$\delta^r$-homogeneous}, if there exists $\mu\in\mathbb{R}$ such that
\begin{equation*}
    V(\delta_\varepsilon^r(\x))=\varepsilon^\mu V(\x)\quad\forall\,\varepsilon>0.
\end{equation*}
A vector function $f=f(\x,\y):\mathbb{R}^{2n}\rightarrow\mathbb{R}^{n}$ is called \textit{$\delta^r$-homogeneous}, if there exists $\mu\in \mathbb{R}$ such that its component $f_i$ is a $\delta^r$-homogeneous function of degree $\mu+r_i,$ i.e.
\begin{equation*}
    f_i(\delta_\varepsilon^r(\x),\delta_\varepsilon^r(\y))=\varepsilon^{\mu+r_i}f_i(\x,\y)\quad\forall\,\varepsilon>0,\quad i=\overline{1,n},
\end{equation*}
where $\x,\y\in\mathbb{R}^n.$
In both cases, the constant $\mu$ is called \textit{the degree of homogeneity}.
It is worth mentioning that the homogeneous norm is a $\delta^r$-homogeneous function of degree one: $\|\delta_\varepsilon^r(\x)\|_{r,p}=\varepsilon\|\x\|_{r,p}.$
Assume that $\mu\geq -\min_{i=\overline{1,n}} r_i.$

\begin{lemma} \label{lemma:bound_f_i}
There exist $m_i>0$ such that the components of the $\delta^r$-homogeneous function $f(\x,\y)$ satisfy  
\begin{equation*}%\label{eq:bound_fi}
    |f_i(\x,\y)|\leq m_i\left(\|\x\|_{r,p}^{\mu+r_i}+\|\y\|_{r,p}^{\mu+r_i}\right),\quad i=\overline{1,n}.
\end{equation*}
\end{lemma}
%\begin{comment}
\begin{proof} If $\mu+r_i>0,$ then we take
$$
m_i=\max_{\|\x\|_{r,p}^{\mu+r_i}+\|\y\|_{r,p}^{\mu+r_i}=1}|f_i(\x,\y)|>0,$$
and $\varepsilon=(\|\x\|_{r,p}^{\mu+r_i}+\|\y\|_{r,p}^{\mu+r_i})^{-1/(\mu+r_i)}.$ It can be easily seen that
$$
\|\delta_\varepsilon^r(\x)\|_{r,p}^{\mu+r_i} + \|\delta_\varepsilon^r(\y)\|_{r,p}^{\mu+r_i}=1.
$$
This implies
$$
    |f_i(\delta_\varepsilon^r(\x),\delta_\varepsilon^r(\y))|\leq m_i.
$$
Using the definition of homogeneity, we arrive at the result. If \mbox{$\mu+r_i=0,$} then the same conclusion can be drawn with
$$
m_i=\max_{\|\x\|_{r,p}^{k}+\|\y\|_{r,p}^{k}=1}|f_i(\x,\y)|>0
$$
for any $k>0.$
\end{proof}
%\end{comment}
\begin{lemma} \label{lemma:bound_der_f_ij}
Assume that $f(\x,\y)$ is continuously differentiable with respect to $\x$ and $\delta^r$-homogeneous. Then, there exist \mbox{$\eta_{ij}>0$} such that
\begin{equation*}%\label{eq:bound_df}
    \left|\frac{\partial f_i(\x,\y)}{\partial \x_j}\right|\leq \eta_{ij}\left(\|\x\|_{r,p}^{\mu+r_i-r_j}+\|\y\|_{r,p}^{\mu+r_i-r_j}\right),\ i,j=\overline{1,n},
\end{equation*}
if $\mu+r_i-r_j>0,$ and 
\begin{equation*}%\label{eq:bound_df}
    \left|\frac{\partial f_i(\x,\y)}{\partial \x_j}\right|\leq \frac{\eta_{ij}}{\left(\|\x\|_{r,p}^{-\mu-r_i+r_j}+\|\y\|_{r,p}^{-\mu-r_i+r_j}\right)},\ i,j=\overline{1,n},
\end{equation*}
if $\mu+r_i-r_j<0$ and
at least one of the vectors $\x$ and $\y$ is nonzero.
\end{lemma}
Now, consider a time delay system of the form
\begin{equation} \label{eq:time-delay_system}
\dot{x}(t)=f(x(t),x(t-h)), 
\end{equation}
where $x(t)\in \mathbb{R}^n,$ $h>0$ is a constant delay. The following assumptions are made.
\begin{assumption}\label{as:hom_degree}
The vector function $f(\textup{x},\textup{y})$ is continuously differentiable with respect to $\x\in\mathbb{R}^n,$ locally Lipshitz with respect to $\y\in\mathbb{R}^n,$ and $\delta^r$-homogeneous of degree $\mu>0.$
\end{assumption}
\begin{assumption}\label{as:delay_free_asymp}
The delay free system
\begin{equation} \label{eq:delay-free_system}
    \dot{x}(t)=f(x(t),x(t))
\end{equation}
is asymptotically stable.
\end{assumption}
In \cite{rosier1992homogeneous}, \cite{zubov1962oxford} the existence of a homogeneous Lyapunov function for system \eqref{eq:delay-free_system} is established. More precisely, it is proven that for any $l\in\mathbb{N}$ and $\gamma\geq l\max_{i=\overline{1,n}}\{r_i\}$ there exists a positive definite $\delta^r$-homogeneous of degree $\gamma$ and of class $C^l$ Lyapunov function $V(\x)$ such that its time derivative with respect to system \eqref{eq:delay-free_system} is a negative definite
$\delta^r$-homogeneous function of degree $\gamma+\mu,$ that is
\begin{equation}
\label{eq:bound_dot_V}
\left(\frac{\partial V(\x)}{\partial \x}\right)^T f(\x,\x)\leq -\w\|\x\|_{r,p}^{\gamma+\mu},\quad \w>0.
\end{equation}
We set $l=2$ and use a Lyapunov function $V(\x)$ of class $C^2$ and the homogeneity degree $\gamma\geq 2\max_{i=\overline{1,n}}\{r_i\}$ below. The following estimates hold \cite{bacciotti2006liapunov}, \cite{zubov1962oxford}:
\begin{gather}\label{eq:bound_V}
\alpha_0\|\x\|_{r,p}^\gamma\leq V(\x)\leq \alpha_1\|\x\|_{r,p}^\gamma,\quad \alpha_0,\alpha_1>0,\\
\label{eq:bound_dV}
\left|\frac{\partial V(\x)}{\partial \x_i}\right|\leq \beta_{i}\|\x\|_{r,p}^{\gamma-r_i},\quad \left|\frac{\partial^2 V(\x)}{\partial \x_i \partial \x_j}\right|\leq \psi_{ij}\|\x\|_{r,p}^{\gamma-r_i-r_j},
\end{gather}
where $\beta_{i},\psi_{ij}>0,$ $i,j=\overline{1,n}.$

It is proved in \cite{aleksandrov2014weighted} that if Assumptions~\ref{as:hom_degree} and \ref{as:delay_free_asymp} hold, then the trivial solution of time delay system \eqref{eq:time-delay_system} is asymptotically stable for all $h>0.$ In the next section, we present the Lyapunov-Krasovskii functional validating this result, and further construct the estimates for the solutions of \eqref{eq:time-delay_system} and for the attraction region. An important step in the obtention of the estimates is the use of a lower bound for the functional on the special set
\begin{multline*}
S_\alpha = \Bigl\{\varphi\in C([-h,0],\mathbb{R}^n)\Bigl\arrowvert\\
\|\varphi(\theta)\|_{r,p}\leq \alpha\|\varphi(0)\|_{r,p},\;\theta\in[-h,0]\Bigr\},
\end{multline*}
where $\alpha>1.$ This set was introduced in the Lyapunov-Krasovskii analysis in \cite{medvedeva2015synthesis} and \cite{alexandrova2018junction} for linear and nonlinear time delay systems, respectively. In particular, it was shown that it is enough to construct the lower bound for the functional on the set $S_\alpha$ in order to prove asymptotic stability.
\section{The Functional Construction}
A natural generalization of the Lyapunov-Krasovskii functional introduced in \cite{Voronezh,alexandrova2019lyapunov} to the case of $\delta^r$-homogeneous time delay systems is
\begin{align}\label{eq:functional}
v(\varphi)&=V(\varphi(0))+\left.\left(\frac{\partial V(\x)}{\partial \x}\right)^T \right|_{\x=\varphi(0)}\int_{-h}^{0}f(\varphi(0),\varphi(\theta))d\theta 
\\
&+\int_{-h}^{0}(\mathrm{w_1}+(h+\theta)\mathrm{w_2})\|\varphi(\theta)\|_{r,p}^{\gamma+\mu}d\theta.\notag
\end{align}
Here, $\mathrm{w}_1, \mathrm{w}_2>0,$ and $\mathrm{w}_0=\mathrm{w}-\mathrm{w}_1-h\mathrm{w}_2>0$. In this section, we show that functional \eqref{eq:functional} satisfies the classical Lyapunov-Krasovskii theorem \cite{hale2013introduction}.
%admits positive definite lower and upper bounds as well as negative definite bound for the derivative along the solutions of system \eqref{eq:time-delay_system}.
For the sake of brevity, the three summands of \eqref{eq:functional} are denoted by $I_1(\varphi)$, $I_2(\varphi)$ and $I_3(\varphi),$ respectively.

\begin{lemma}
\label{lemma:lower_bound}
There exist $\delta>0$ and $a_1,\,a_2>0$ such that
\begin{equation} \label{eq:lower_bound}
    v(\varphi)\geq a_1\|\varphi(0)\|_{r,p}^\gamma+a_2\int_{-h}^{0}\|\varphi(\theta)\|_{r,p}^{\gamma+\mu} d\theta
\end{equation}
in the neighbourhood $\|\varphi\|_\mathscr{H}\leq \delta.$
\end{lemma}
\begin{proof}
It is obvious that
\begin{equation*}
    I_1(\varphi)\geq \alpha_0\|\varphi(0)\|_{r,p}^\gamma,\;\, I_3(\varphi)\geq \mathrm{w}_1\int_{-h}^{0}\|\varphi(\theta)\|_{r,p}^{\gamma+\mu}d\theta.
\end{equation*}
Now, we use Lemma~\ref{lemma:bound_f_i} and the first of bounds \eqref{eq:bound_dV} for the second summand of the functional:
\begin{align*}
%|I_2(\varphi)|=\left|\sum_{i=1}^n \left.\frac{\partial V(\x)}{\partial \x_i} \right|_{\x=\varphi(0)}\int_{-h}^{0}f_i(\varphi(0),\varphi(\theta))d\theta\right|
%\leq \sum_{i=1}^n \beta_i m_i\|\varphi(0)\|_{r,p}^{\gamma-r_i}
%\int_{-h}^{0}\left(\|\varphi(0)\|_{r,p}^{\mu+r_i}+\|\varphi(\theta)\|_{r,p}^{\mu+r_i}\right)d\theta
|I_2(\varphi)|
&\leq h\left(\sum_{i=1}^n \beta_i m_i\right)\|\varphi(0)\|_{r,p}^{\gamma+\mu} + \sum_{i=1}^n\beta_i m_i\chi^{\gamma-\mu-2r_i}\\
&\times\left(\frac{\|\varphi(0)\|_{r,p}}{\chi}\right)^{\gamma-r_i}
\int_{-h}^{0}(\chi\|\varphi(\theta)\|_{r,p})^{\mu+r_i}d\theta,
\end{align*}
where $\chi>0$ is a parameter. Using the standard inequality $A^{p_1} B^{p_2}\leq A^{p_1+p_2}+B^{p_1+p_2},$ where $p_1,p_2\geq0$ and $A,B\geq0$, we get
\begin{align*}
|I_2(\varphi)|&\leq h\left(\sum_{i=1}^n \beta_i m_i \left(1+\chi^{-2(\mu+r_i)}\right)\right)\|\varphi(0)\|_{r,p}^{\gamma+\mu}\\&+\left(\sum_{i=1}^n\beta_i m_i\chi^{2(\gamma-r_i)}\right)\int_{-h}^{0}\|\varphi(\theta)\|_{r,p}^{\gamma+\mu}d\theta.
\end{align*}
Combining all summands together and making use of $I_2(\varphi)\geq -|I_2(\varphi)|$ along with $\|\varphi\|_\mathscr{H}\leq \delta,$ we arrive at the lower bound \eqref{eq:lower_bound} with
\begin{align*}
    a_1&=\alpha_0-h\sum_{i=1}^n\beta_i m_i\left(1+\chi^{-2(\mu+r_i)}\right)\delta^{\mu},\\
    a_2&=\mathrm{w}_1-\sum_{i=1}^n\beta_i m_i\chi^{2(\gamma-r_i)}.
\end{align*}
Here, the parameter  $\chi>0$ is chosen such that $a_2>0,$ and
\begin{equation} \label{H1}
    \delta<H_1 =\left(\frac{\alpha_0}{h\sum_{i=1}^n\beta_i m_i\left(1+\chi^{-2(\mu+r_i)}\right)}\right)^{1/\mu}. 
\end{equation}
%Making use of
%$I_2(\varphi)\geq -|I_2(\varphi)|$ and combining all summands together, we arrive at the lower bound
%\begin{gather*}
%v(\varphi)\geq \left(\alpha_0-h\sum_{i=1}^n\beta_i m_i\left(1+\chi^{-2(\mu+r_i)}\right)\|\varphi(0)\|_{r,p}^{\mu}\right) 
%\\
%\times \|\varphi(0)\|_{r,p}^\gamma +\left(\mathrm{w}_1-\sum_{i=1}^n\beta_im_i\chi^{2(\gamma-r_i)}\right)\int_{-h}^{0}\|\varphi(\theta)\|_{r,p}^{\gamma+\mu} d\theta.
%\end{gather*}
%Now, choose $\chi>0$ such that $\mathrm{w}_1-\sum_{i=1}^n\beta_im_i\chi^{2(\gamma-r_i)}>0,$ and define
%\begin{equation} \label{H1}
%    H_1 =\left(\frac{\alpha_0}{h\sum_{i=1}^n\beta_i m_i\left(1+\chi^{-2(\mu+r_i)}\right)}\right)^{1/\mu}. 
%\end{equation}
%Finally, we obtain bound \eqref{eq:lower_bound} with
%\begin{align*}
%    a_1&=\alpha_0-h\sum_{i=1}^n\beta_i m_i\left(1+\chi^{-2(\mu+r_i)}\right)\delta^{\mu},\\
%    a_2&=\mathrm{w}_1-\sum_{i=1}^n\beta_i m_i\chi^{2(\gamma-r_i)}
%\end{align*}
%for an arbitrary positive $\delta<H_1.$
\end{proof}

\begin{lemma}
\label{lemma:derivative}
There exist $\delta>0$ and $c_0,c_1,c_2>0$ such that the time derivative of functional~\eqref{eq:functional} along the solutions of system \eqref{eq:time-delay_system} satisfies
\begin{align} \label{eq:derivative}
\frac{\mathrm{d}v(x_t)}{\mathrm{d}t}&\leq -c_0\|x(t)\|_{r,p}^{\gamma+\mu}-c_1\|x(t-h)\|_{r,p}^{\gamma+\mu}\\&-c_2\int_{-h}^{0}\|x(t+\theta)\|_{r,p}^{\gamma+\mu}d\theta,\quad\text{if}\quad \|x_t\|_\mathscr{H}\leq \delta.\notag
\end{align}
\end{lemma}
\begin{proof}
Similarly to the case of standard dilation \cite{alexandrova2019lyapunov}, we differentiate the functional along the solutions of system \eqref{eq:time-delay_system}:
\begin{gather*}
\begin{split}
    \frac{\mathrm{d}v(x_t)}{\mathrm{d}t} 
    &= -\mathrm{w}_0\|x(t)\|_{r,p}^{\gamma+\mu}-\mathrm{w}_1\|x(t-h)\|_{r,p}^{\gamma+\mu}\\
    &-\mathrm{w}_2\int_{-h}^{0}\|x(t+\theta)\|_{r,p}^{\gamma+\mu}d\theta + \sum_{j=1}^{2}\Lambda_j,\quad\text{where}
    \end{split}\\
    \begin{split}
    \Lambda_1&=\sum_{i,j=1}^n \left.\frac{\partial V(\x)}{\partial \x_i} \right|_{\x=x(t)}\left.\int_{t-h}^{t}\frac{\partial f_i(\x,x(s))}{\partial \x_j}\right|_{\x=x(t)} \mathrm{d}s\\
&\times f_j(x(t),x(t-h)),\quad
\Lambda_2=\sum_{i,j=1}^n f_i(x(t),x(t-h))\\&\times\left.\left(\frac{\partial^2 V(\x)}{\partial \x_i\x_j}\right)\right|_{\x=x(t)}\int_{-h}^{0}f_j(x(t),x(t+\theta))\mathrm{d}\theta.
\end{split}
\end{gather*}
%\begin{align*}
%\Lambda_1&=\left.\left(\frac{\partial V(\x)}{\partial \x}\right)^T \right|_{\x=x(t)}\left.\int_{t-h}^{t}\frac{\partial f(\x,x(s))}{\partial \x}\right|_{\x=x(t)} \mathrm{d}s\\
%&\times f(x(t),x(t-h)),\quad
%\Lambda_2=(f(x(t),x(t-h)))^T\\&\times\left.\left(\frac{\partial^2 V(\x)}{\partial \x^2}\right)\right|_{\x=x(t)}\int_{-h}^{0}f(x(t),x(t+\theta))\mathrm{d}\theta.
%\end{align*}
Next, we estimate $\Lambda_1$ and $\Lambda_2$ with the help of Lemmas~\ref{lemma:bound_f_i}, \ref{lemma:bound_der_f_ij} and inequalities \eqref{eq:bound_dV}. We introduce the sets of indices
\begin{align*}
    R_1&=\{(i,j):\;i,j=\overline{1,n},\;\mu+r_i-r_j\geq 0\},\\
    R_2&=\{(i,j):\;i,j=\overline{1,n},\;\mu+r_i-r_j< 0\}
\end{align*}
for the estimation of $\Lambda_1.$ Notice that Lemma~\ref{lemma:bound_der_f_ij} implies that
$$
\left|\frac{\partial f_i(\x,\y)}{\partial \x_j}\right|\leq \eta_{ij},\quad (i,j)\in R_2,
$$
hence,
\begin{gather*}
 \Lambda_1\leq \!\!\!\sum_{(i,j)\in R_1}\!\!\!\beta_i m_j \|x(t)\|_{r,p}^{\gamma-r_i}(\|x(t)\|_{r,p}^{\mu+r_j}+\|x(t-h)\|_{r,p}^{\mu+r_j})\\
 \times\int_{-h}^{0}\eta_{ij}(\|x(t)\|_{r,p}^{\mu+r_i-r_j}+\|x(t+\theta)\|_{r,p}^{\mu+r_i-r_j})d\theta\\
 +\!\!\!\!\!\sum_{(i,j)\in R_2}\!\!\!\!\! h\beta_i m_j \eta_{ij} \|x(t)\|_{r,p}^{\gamma-r_i}(\|x(t)\|_{r,p}^{\mu+r_j}+\|x(t-h)\|_{r,p}^{\mu+r_j}),\\
 \Lambda_2\leq \sum_{i=1}^n\sum_{j=1}^n m_i m_j \psi_{ij}(\|x(t)\|_{r,p}^{\mu+r_i}+\|x(t-h)\|_{r,p}^{\mu+r_i})
 \end{gather*}
 \begin{gather*}
 \times\|x(t)\|_{r,p}^{\gamma-r_i-r_j}\int_{-h}^{0}(\|x(t)\|_{r,p}^{\mu+r_j}+\|x(t+\theta)\|_{r,p}^{\mu+r_j})d\theta.
 \end{gather*}
Using the standard inequality $$A^{p_1} B^{p_2} C^{p_3}\leq A^{p_1+p_2+p_3}+B^{p_1+p_2+p_3}+C^{p_1+p_2+p_3},$$ where $p_1,p_2,p_3\geq0$ and $A,B,C\geq0$, and defining
\begin{align*}
s_{ij}&=\left\{
\begin{aligned}
&1,&\quad &(i,j)\in R_1,\\
&\delta^{r_j-r_i-\mu},&\quad &(i,j)\in R_2,
\end{aligned}\right.\\
L&=\sum_{i=1}^n\sum_{j=1}^n m_j\left(\beta_i\eta_{ij}s_{ij}+m_i\psi_{ij}\right),\\
g(x_t)&=4h\|x(t)\|_{r,p}^{\gamma+2\mu}+2h\|x(t-h)\|_{r,p}^{\gamma+2\mu}\\
&+2\int_{-h}^{0}\|x(t+\theta)\|_{r,p}^{\gamma+2\mu},
\end{align*}
we arrive at
$
\Lambda_1+\Lambda_2 \leq L g(x_t).
$
%\begin{align*}
%\Lambda_1&\leq  \sum_{i=1}^n\sum_{j=1}^n\beta_i m_j\eta_{ij} s_{ij}g(x_t),\\
%\Lambda_2&\leq \sum_{i=1}^n\sum_{j=1}^n m_i m_j \psi_{ij}g(x_t).
%\end{align*}
Since $\mu>0,$ we obtain bound \eqref{eq:derivative} with
\begin{align*}
    c_0=\mathrm{w}_0-4hL\delta^\mu,\;\, c_1=\mathrm{w}_1-2hL\delta^\mu,\;\,c_2=\mathrm{w}_2-2L\delta^\mu.
\end{align*}
It is enough to choose
\begin{equation} \label{H2}
\delta<H_2=\left(\min\left\{\frac{\mathrm{w}_0}{4hL},\frac{\mathrm{w}_1}{2hL},\frac{\mathrm{w}_2}{2L}\right\}\right)^{1/\mu}
\end{equation}
to end the proof.
\end{proof}

\begin{lemma}
\label{lemma:upper_bound}
There exist $b_1,b_2>0$ such that
\begin{equation} \label{eq:upper_bound}
    v(\varphi)\leq b_1\|\varphi(0)\|_{r,p}^{\gamma}+b_2\int_{-h}^{0}\|\varphi(\theta)\|_{r,p}^{\gamma}d\theta,
\end{equation}
if $\|\varphi\|_\mathscr{H}\leq \delta.$
\end{lemma}
\begin{proof}
The bound is obtained straightforwardly with
\begin{align*}
    b_1&=\alpha_1+2h\left(\sum_{i=1}^n \beta_i m_i\right)\delta^\mu,\\ b_2&=\left(\mathrm{w}_1+h\mathrm{w}_2+\sum_{i=1}^n\beta_i m_i\right)\delta^{\mu}.
\end{align*}
\end{proof}
\begin{corollary}
Functional \eqref{eq:functional} admits an upper bound
\begin{gather}\label{eq:sec_upper_bound_v}
    %v(\varphi)\leq (b_1+h b_2)\|\varphi\|_h^\gamma,\\
    v(\varphi)\leq \alpha_1\|\varphi(0)\|^\gamma + b_3 \|\varphi\|_\mathscr{H}^{\gamma+\mu},
\end{gather}
where $$b_3=\left(\mathrm{w}_1+h\mathrm{w}_2+2h\sum_{i=1}^n \beta_im_i\right)h.$$
\end{corollary}

Now, we present a lower bound for the functional $v(\varphi)$ on the set $S_{\alpha}.$ This bound  will be useful for the construction the of the estimates in Sections~IV and V.
\begin{lemma}
\label{lemma:lower_bound_set_S}
There exist $\delta>0$ and $\tilde{a}_1=\tilde{a}_1(\alpha)>0$ such that
\begin{equation} \label{eq:lower_bound_set_S}
    v(\varphi)\geq \tilde{a}_1\|\varphi(0)\|_{r,p}^\gamma+\mathrm{w}_1\int_{-h}^{0}\|\varphi(\theta)\|_{r,p}^{\gamma+\mu} d\theta,
\end{equation}
if $\varphi\in S_\alpha$ and $\|\varphi\|_\mathscr{H}\leq \delta.$
\end{lemma}
\begin{proof}
Using $\|\varphi(\theta)\|_{r,p}\leq \alpha\|\varphi(0)\|_{r,p},$ \mbox{$\theta\in[-h,0],$} for the estimation of the second summand, we obtain that bound \eqref{eq:lower_bound_set_S} is satisfied with
\begin{gather}
\tilde{a}_1=\alpha_0-h\left(\sum_{i=1}^n (1+\alpha^{\mu+r_i})m_i\beta_i\right) \delta^\mu,\quad\text{where} \notag
\\
 \label{H3}
    \delta<H_3=\left(\frac{\alpha_0}{h\sum_{i=1}^n (1+\alpha^{\mu+r_i})m_i\beta_i}\right)^{1/\mu}.
\end{gather}
\end{proof}

\section{Estimates for the Attraction Region}
%In this section, we provide and alternative proof to the stability result achieved in the Razumikhin framework (Theorem 1 in \cite{aleksandrov2014delay}). The new proof is possible thank to the above Lyapunov-Krasovskii functional.
%Summarizing, we have proved that the functional $v(\varphi)$ admits the lower bounds \eqref{eq:sec_upper_bound_v} and \eqref{eq:lower_bound_set_S}, the upper bound \eqref{eq:lower_bound_v}, and that its time derivative along the solutions of system \eqref{eq:time-delay_system} is such that \eqref{eq:sec_bound_dot_v} holds. Functional \eqref{eq:functional} meets the conditions \eqref{ccv} and \eqref{cdv} of Theorem \ref{th:LK}, as well as all conditions of Theorem~\ref{th:LK(A&Z}. We conclude, as previously done in the Razumikhin framework (Theorem 1 in \cite{aleksandrov2014delay}), that  the trivial solution of the homogeneous system with delay \eqref{eq:time-delay_system} is asymptotically stable for all delay $h>0$ and $\mu>0$. \\
Lemmas~\ref{lemma:lower_bound}--\ref{lemma:lower_bound_set_S} allow us to present straightforwardly estimates of the domain of attraction of the trivial solution of system \eqref{eq:time-delay_system}. The proofs are very similar to that in \cite{alexandrova2019lyapunov} (see Corollary~10 and Remark~11). The estimates differ in the lower bound for the functional used: bound \eqref{eq:lower_bound} in Theorem~\ref{thm:attr_region_LK_w} and bound \eqref{eq:lower_bound_set_S} in Theorem~\ref{thm:attr_region_LK_set_S}.

\begin{theorem}\label{thm:attr_region_LK_w}
Let $\Delta$ be a positive root of equation
\begin{equation*}
\alpha_1\Delta^\gamma + b_3 \Delta^{\gamma+\mu} = a_1\delta^\gamma,
\end{equation*}
where $\delta<\min\{H_1,H_2\},$ and $H_1$ and $H_2$ are defined by \eqref{H1} and \eqref{H2}, respectively.
If system \eqref{eq:delay-free_system} is asymptotically stable, then the set of initial functions
\begin{equation*} %\label{eq:attraction_LK}
\Omega=\{\varphi\in C([-h,0],\mathbb{R}^n):\|\varphi\|_\mathscr{H}<\Delta\},  \end{equation*}
estimates the attraction region of the trivial solution of \eqref{eq:time-delay_system}.
\end{theorem}

\begin{theorem}\label{thm:attr_region_LK_set_S}
Let $\Delta_\alpha$ be a positive root of equation
\begin{equation*}
\alpha_1\Delta_\alpha^\gamma + b_3 \Delta_\alpha^{\gamma+\mu} = \tilde{a}_1\delta^\gamma,
\end{equation*}
where $\delta<\min\{H_2,H_3\},$ and $H_3$ is defined by \eqref{H3}.
If system \eqref{eq:delay-free_system} is asymptotically stable, then the set of initial functions
\begin{equation*} %\label{eq:attraction_LK}
\Omega_\alpha=\{\varphi\in C([-h,0],\mathbb{R}^n):\|\varphi\|_\mathscr{H}<\Delta_\alpha\},  
\end{equation*}
estimates the attraction region of the trivial solution of \eqref{eq:time-delay_system}.
\end{theorem}

\begin{remark}
Proofs of Theorem~\ref{thm:attr_region_LK_w} and Theorem~\ref{thm:attr_region_LK_set_S} imply that $\|x(t,\varphi)\|_{r,p}<\delta$ $\,\forall\, t\geq 0,$ if $\|\varphi\|_\mathscr{H}<\Delta$ or $\|\varphi\|_\mathscr{H}<\Delta_\alpha.$
\end{remark}

\section{Estimates for the solutions}
For the standard dilation, the estimates for the solutions obtained with the help of the Lyapunov-Krasovskii functional \eqref{eq:functional} are presented in \cite{portilla2020comparison}. We extend straightforwardly this result to the case of weighted dilation in Section~\ref{subs_classical}. A novel approach which combines the use of functional \eqref{eq:functional} with ideas of the Razumikhin framework is presented in Section~\ref{subs_via_set_S}.

\subsection{Classical Approach}
\label{subs_classical}
Bounds \eqref{eq:derivative} and \eqref{eq:upper_bound} imply that if $\|x_t\|_\mathscr{H}\leq \delta,$ $t\geq0,$ then
\begin{align}
\label{eq:sec_bound_dot_v}
\frac{dv(x_t)}{dt}&\leq -c\left(\|x(t)\|_{r,p}^{\gamma+\mu}+\int_{-h}^{0}\|x(t+\theta)\|_{r,p}^{\gamma+\mu}d\theta\right),\\
\label{eq:first_upper_bound_v}
    v(x_t)&\leq b\left(\|x(t)\|_{r,p}^{\gamma}+\int_{-h}^{0}\|x(t+\theta)\|_{r,p}^{\gamma}d\theta\right),
\end{align}
where $c=\min\{c_0,c_2\},$ $b=\max\{b_1,b_2\},$ \mbox{$\delta<\min\{H_1,H_2\}.$} Define the values
$$
\rho_1=\bigl(2\max\{1,h\}\bigr)^{\frac{\mu}{\gamma}},\quad\rho_2=\frac{c}{\rho_1 b^{\frac{\gamma+\mu}{\gamma}}}.
$$
The following relation was established in \cite{portilla2020comparison} on the basis of Hölder's inequality:
\begin{multline} \label{inqx_w}
    \left(\|x(t)\|_{r,p}^{\gamma}+\int_{-h}^{0}\|x(t+\theta)\|_{r,p}^{\gamma}d\theta\right)^{\frac{\gamma+\mu}{\gamma}}\\ \leq  \rho_1\left(\|x(t)\|_{r,p}^{\gamma+\mu}+\int_{-h}^{0}\|x(t+\theta)\|_{r,p}^{\gamma+\mu}d\theta\right).
\end{multline}
Combining \eqref{eq:sec_bound_dot_v}, \eqref{eq:first_upper_bound_v} and \eqref{inqx_w} gives the following
connection between functional \eqref{eq:functional} and its derivative.
\begin{lemma} \label{lemma:connect_dv_v_w}
The following inequality is satisfied:
\begin{equation}\label{eq:connect_v_dot_v} 
    \frac{dv(x_t)}{dt} \leq -\rho_2 v^{\frac{\gamma+\mu}{\gamma}}(x_t),\quad t\geq0, 
\end{equation}
along the solutions of system \eqref{eq:time-delay_system} with $\|x_t\|_\mathscr{H}\leq \delta.$
\end{lemma}
Considering the comparison equation \cite{khalil1996nonlinear} of the form
\begin{equation}\label{eq:equ_comparison}
    \frac{du(t)}{dt} = -\rho_2 u^{\frac{\gamma+\mu}{\gamma}}
(t),
\end{equation}
with initial condition
$
   u(0)=u_0=(\alpha_1 + b_3 \Delta^{\mu})\|\varphi\|_\mathscr{H}^\gamma
$
and exploiting the classical ideas, we arrive at the following estimates for solutions in the homogeneous norm.
\begin{theorem}\label{th:estimation}
Let Assumptions~\ref{as:hom_degree} and \ref{as:delay_free_asymp} hold. Then, the solutions of system \eqref{eq:time-delay_system} with initial functions with $\|\varphi\|_\mathscr{H}<\Delta,$ where $\Delta$ is defined in Theorem~\ref{thm:attr_region_LK_w}, admit an estimate of the form 
\begin{gather} \label{eq:estimate_final}
   \|x(t,\varphi)\|_{r,p} \leq \frac{\hat{c}_1 \|\varphi\|_\mathscr{H}}{\left(1+\hat{c}_2\|\varphi\|_\mathscr{H}^{\mu}t\right)^{1/\mu}},\quad\text{where}
\\
\begin{split}
\hat{c}_1&=\left(\frac{\alpha_1+b_3 \Delta^{\mu}}{a_1}\right)^\frac{1}{\gamma}=\frac{\delta}{\Delta},\notag
\\
\hat{c}_2&=\frac{c}{b}\left(\frac{\mu}{\gamma}\right)\left(\frac{\alpha_1+b_3 \Delta^{\mu}}{2b\max\{1,h\}}\right)^{\frac{\mu}{\gamma}}.\notag
\end{split}
\end{gather}
\end{theorem}
\begin{remark}
Note that $\|\varphi\|_\mathscr{H}<\Delta$ implies $\|x_t(\varphi)\|_\mathscr{H}<\delta$ for all $t\geq0$ according to Theorem~\ref{thm:attr_region_LK_w} thus making use of \eqref{eq:sec_bound_dot_v} and \eqref{eq:first_upper_bound_v} legal.
\end{remark}
\subsection{Novel Approach Using the Set $S_\alpha$}
\label{subs_via_set_S}
The lower bound (\ref{eq:lower_bound_set_S}) is expected to be less conservative than the original bound (\ref{eq:lower_bound}), since it should hold on the reduced set $S_{\alpha}$ instead of the set of all continuous functions.
Thus, a natural question appears: Can we replace the constant $a_1$  coming from the lower bound in \eqref{eq:estimate_final} with $\tilde{a}_1$? We give an affirmative answer to this question with some restrictions below.
Exploring the proof of Theorem~\ref{th:estimation}, one finds that the lower bound for the functional is used at the final step of the proof, namely, $$a_1\|x(t,\varphi)\|_{r,p}^\gamma\leq v(x_t)\leq u(t),$$
where $u(t)$ is the solution of \eqref{eq:equ_comparison}. It is crucial that the last formula is true for all solutions, not only those with each segment in $S_\alpha.$ A similar difficulty appears while constructing the estimates of solutions via the Razumikhin theorem \cite{aleksandrov2012asymptotic}. Here, we adapt the ideas of \cite{aleksandrov2012asymptotic} to reduce the conservatism of Theorem~\ref{th:estimation}.
%(i.e. to choose parameters of the comparison equation in such a way that its solution satisfies the Razumikhin condition)

We start with Lemma~\ref{lemma:connect_dv_v_w}. To overcome the mentioned difficulty, we take $\rho<\rho_2,$ and consider the comparison initial value problem
\begin{gather}
    \frac{du(t)}{dt} = -\rho u^{\frac{\gamma+\mu}{\gamma}} (t),\\
    u(0)=\tilde{u}_0=(\alpha_1 + b_3 \Delta_\alpha^{\mu})\|\varphi\|_\mathscr{H}^\gamma,
\end{gather}
which admits the solution
\begin{equation*}
    u(t) = \tilde{u}_0\left[1+\rho\left(\frac{\mu}{\gamma}\right)\tilde{u}_0^{\frac{\mu}{\gamma}}t\right]^{-\frac{\gamma}{\mu}}.
\end{equation*}
Now, we present a set of auxiliary results to extend the approach of \cite{aleksandrov2012asymptotic} to the Lyapunov-Krasovskii framework. In Lemma~\ref{lemma:add_restric}, a choice of $\rho$ is made. Such choice is always possible due to $\alpha>1.$   Theorem~\ref{thm_set_S_prelim} allows us to switch from the bound on the set $S_\alpha$ to the bound which holds for all solutions in a certain neighbourhood of the trivial one. The proofs are omitted due to length limitations.
\begin{lemma}\label{lemma:strict_comparison}
If $\|\varphi\|_\mathscr{H}<\Delta_\alpha,$ then
\begin{equation*}
%\label{strict_comparison}
v(x_t)<u(t),\quad t\geq 0.
\end{equation*}
\end{lemma}
%Next, we choose $\rho$ in accordance with the following lemma. Such choice is always possible because $\alpha>1.$ 
\begin{lemma} \label{lemma:add_restric}
If the condition
\begin{equation}
\label{add_restrict}
1+\rho h\left(\frac{\mu}{\gamma}\right)(\alpha_1 +b_3\Delta_\alpha^{\mu})^{\frac{\mu}{\gamma}}\Delta_\alpha^{\mu}\leq \alpha^{\mu}
\end{equation}
holds, then $u(t+\theta)<\alpha^\gamma u(t)$
%\begin{equation}
%\label{Raz_cond_u}
%u(t+\theta)<\alpha^\gamma u(t)
%\end{equation}
for all $t\geq 0$ and $\theta\in[-h,0]$ such that $t+\theta\geq 0.$
\end{lemma}
%The following theorem provides a key step for the proof of the main result. Although the lower bound  for the functional on the set $S_\alpha$ is used, the result holds for all solutions in a certain neighbourhood of the trivial one.
\begin{theorem} \label{thm_set_S_prelim}
If $\|\varphi\|_\mathscr{H}<\Delta_\alpha$ and inequality \eqref{add_restrict} holds, then
$$
\tilde{a}_1\|x(t,\varphi)\|_{r,p}^\gamma< u(t),\quad t\geq 0.
$$
\end{theorem}
%Now we are ready to present the main result of the section.
Based on Lemmas~\ref{lemma:strict_comparison}, \ref{lemma:add_restric} and Theorem~\ref{thm_set_S_prelim} we present the main result of the section.
\begin{theorem}\label{th:estimation_S}
Let Assumptions~\ref{as:hom_degree}, \ref{as:delay_free_asymp} and inequality \eqref{add_restrict} hold. Then, the solutions of system \eqref{eq:time-delay_system} corresponding to the initial functions with $\|\varphi\|_\mathscr{H}<\Delta_\alpha,$ where $\Delta_\alpha$ is defined in Theorem~\ref{thm:attr_region_LK_set_S}, admit an estimate of the form \eqref{eq:estimate_final}
with
\begin{align*}
\hat{c}_1&=\left(\frac{\alpha_1+b_3 \Delta_\alpha^{\mu}}{\tilde{a}_1}\right)^\frac{1}{\gamma}=\frac{\delta}{\Delta_\alpha},
\\
\hat{c}_2&=\frac{c}{b}\left(\frac{\mu}{\gamma}\right)\left(\frac{\alpha_1+b_3 \Delta_\alpha^{\mu}}{2b\max\{1,h\}}\right)^{\frac{\mu}{\gamma}}.
\end{align*}
\end{theorem}

\section{ILLUSTRATIVE EXAMPLE}
%%%% For the journal version --- place back comparison between p=1 and p=5, place back table with discussion, may be pictures for p=1 and p=5
%%%%
Consider the following system, which is
used to model complex interactions, either instantaneous or delayed, occurring amongst transcription
factors and target genes  \cite{efimov2016global}:
\begin{gather}
\begin{split} \label{example}
\dot{x}_1(t)&=-\kappa_1x_1^2(t)+\lambda_1x_2(t-h),
\\
\dot{x}_2(t)&=-\kappa_2x_2^{3/2}(t)+\lambda_2x_2(t)x_1(t-h).
\end{split}
\end{gather}
Here $x_1(t),x_2(t)\in\mathbb{R}^+$ represent interactions occurring in a genetic
network, $h>0$ is the transition delay in the network, and $\kappa_1,\kappa_2,\lambda_1,\lambda_2$ are positive parameters. System \eqref{example} is $\delta^r$-homogeneous of degree $\mu=1$ with $(r_1,r_2)=(1,2).$ Set $\gamma=4$ and consider the Lyapunov function 
\begin{equation*}%\label{lyap_example}
   V(x)=x_1^4+x_2^2,   
\end{equation*}
which is positive definite. Its derivative along the trajectories of system \eqref{example} when $h=0$ is of the form 
\begin{gather*}
    \dot{V}(x)=-4\kappa_1x_1^5+4\lambda_1x_1^3x_2-2\kappa_2x_2^{5/2}+2\lambda_2x_2^2x_1\\
    \leq -2\min\{2\kappa_1,\kappa_2\}(x_1^5+x_2^{5/2})+4\max\{2\lambda_1,\lambda_2\}\|x(t)\|_{r,p}^5.
\end{gather*}
%This function is $\delta^r$-homogeneous of degree $\gamma+\mu=5$ and admits the upper bound of the form
%\begin{equation*}
%   -2\min\{2\kappa_1,\kappa_2\}(x_1^5+x_2^{5/2})+4\max\{2\lambda_1,\lambda_2\}\|x(t)\|_{r,p}^5.
%\end{equation*}
%Choosing $p=1$ for the homogeneous norm, we arrive at bound \eqref{eq:bound_dot_V} with
%$
%\mathrm{w}=0.125\min\{2\kappa_1,\kappa_2\}-4\max\{2\lambda_1,\lambda_2\}.
%$
%For $p=5$ we obtain
%$
%\mathrm{w}=2\min\{2\kappa_1,\kappa_2\}-4\max\{2\lambda_1,\lambda_2\}.
%$
%Hence, the stability condition for delay-free system corresponding to \eqref{example}, i.e. $\mathrm{w}>0,$ is less conservative with $p=5.$
Choosing $p=5$ for the homogeneous norm, we arrive at bound \eqref{eq:bound_dot_V} with
$\mathrm{w}=2\min\{2\kappa_1,\kappa_2\}-4\max\{2\lambda_1,\lambda_2\}.$
Compute the other constants:
$m_1=\max\{\kappa_1,\lambda_1\},$  $m_2=\kappa_2+\lambda_2,$  $\eta_{11}=2\kappa_1,$ $\eta_{12}=\eta_{21}=0,$  $\eta_{22}=\max\{1.5\kappa_2,\lambda_2\},$  $\beta_1=4,$
$\beta_2=2,$ $\psi_{11}=12,$ $\psi_{12}=\psi_{21}=0,$ $\psi_{22}=2,$ $\alpha_0=1$ and $\alpha_1=2^{1/5}$.
%For $p=1$ we have $\alpha_0=1/8$ and $\alpha_1=1$, and for $p=5$, $\alpha_0=1$ and $\alpha_1=2^{1/5}$.
%For a given set of system parameters $(\kappa_1,\kappa_2,\lambda_1,\lambda_2)=(9,18,0.25,0.5),$ the constants characterizing the estimates for solutions obtained via Theorems~\ref{th:estimation} and \ref{th:estimation_S} are shown in Table~\ref{table_comparison}.
\begin{figure}[H]
     \centering
      \includegraphics[width=0.5\textwidth]{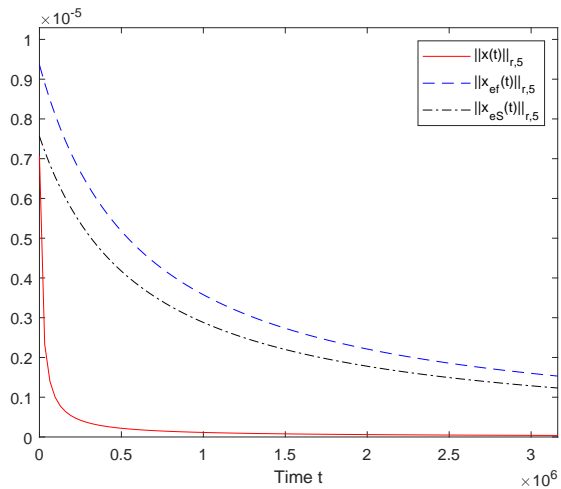}
      \caption{Estimation of the solution of system \eqref{example} for $p=5$}
      \label{figure_estimation_w_1_2}
  \end{figure}
For the parameters $(\kappa_1,\kappa_2,\lambda_1,\lambda_2)=(9,18,0.25,0.5)$ and the initial function $\varphi(\theta)=[5\cdot 10^{-11},5\cdot 10^{-11}]$, $\theta \in [-10,0],$ the system response (continuous line) and the estimates using Theorem~\ref{th:estimation} (dashed line) and  Theorem~\ref{th:estimation_S} (dashed-dot line) are depicted on Figure~\ref{figure_estimation_w_1_2}. We conclude that the use of the set $S_\alpha$ allows us to obtain a tighter estimate than those via the classical approach.
%\begin{table*}[htb]
%\caption{Constants for the estimates of solutions of Example}
%\label{table_example}
%\begin{center}
%\begin{tabular}{cccccccccccc}
%\hline
%$p$ & $\delta$ & $H_1$ & $H_2$ & $\Delta$ & $\hat{c}_1$ & $\hat{c}_2$ & $\chi$ & $\mathrm{w}$ & $\mathrm{w}_0$ & $\mathrm{w}_1$ & $\mathrm{w}_2$\\
%\hline
%$1$ & $1\cdot 10^{-9}$ & $1.6\cdot 10^{-8}$ & $1.8\cdot 10^{-7}$ & $9.8\cdot 10^{-10}$ & $1.0162$ & $0.002$ & $0.135$ & $0.25$ & $0.07$ & $0.012$ & $0.016$\\
%\hline
%$5$ & $1\cdot 10^{-5}$ & $1.6\cdot 10^{-5}$ & $2.5\cdot 10^{-5}$ & $7.5\cdot 10^{-6}$ & $1.3$ & $0.2282$ & $0.44$ & $34$ & $9.63$ & $1.7$ & $2.26$\\
%\hline
%\end{tabular}
%\end{center}
%\end{table*}
%%%%%LAST VERSION OF THE TABLE
%%%
%%%

\section{Conclusion}
In this paper, we present a Lyapunov-Krasovskii functional for weighted homogeneous time delay systems of positive degree and show its potential as an analysis and design tool by computing the estimates of the domain of attraction and of the system solutions.

\bibliographystyle{IEEEtran}
\bibliography{IEEEabrv,mybibliography}

\begin{comment}
\begin{IEEEbiography}{Gerson Portilla}
Biography text here.
\end{IEEEbiography}

\begin{IEEEbiography}{Irina Alexandrova}
Biography text here.
\end{IEEEbiography}

\begin{IEEEbiography}{Sabine Mondié}
Biography text here.
\end{IEEEbiography}
\end{comment}

\end{document}